\documentclass[10pt,aps,pra,twocolumn,superscriptaddress]{revtex4-1}
\usepackage{amsmath}
\usepackage{amsfonts}
\usepackage{amssymb, amsthm}
\usepackage{stmaryrd}
\usepackage{expdlist}
\usepackage{mathtools}
\usepackage{multirow}
\usepackage{verbatim}
\usepackage{alltt}
\usepackage{moreverb}
\usepackage{graphicx,color,graphics}
\usepackage{hyperref}
\usepackage{setspace}
\usepackage{braket}
\usepackage{calligra}
\usepackage{bm}
\usepackage[normalem]{ulem}

\usepackage{amsmath,lipsum}    % need for subequations
\usepackage{graphicx}   % need for figures
\usepackage{verbatim}   % useful for program listings
\usepackage{color}      % use if color is used in text
\usepackage{subfigure}  % use for side-by-side figures
\usepackage{hyperref}   % use for hypertext links, including those to external documents and URLs
\usepackage{setspace,enumerate,amssymb,amsthm,verbatim}
\usepackage{mathtools}
\usepackage[normalem]{ulem}
\usepackage{natbib}
\usepackage{xr} 

\newtheorem{theorem}{Theorem}
\def\one{{\mathchoice {\rm 1\mskip-4mu l} {\rm 1\mskip-4mu l} {\rm
1\mskip-4.5mu l} {\rm 1\mskip-5mu l}}}

\providecommand{\ignore}[1]{}

\newif\ifcmnt
\cmnttrue
\ifdefined\cmntsoff\cmntfalse\fi

\ifcmnt
    \providecommand{\aucmnt}[1]{#1}

\else
    \providecommand{\aucmnt}[1]{}

\fi

\DeclareMathOperator{\Tr}{Tr}

\newcommand{\GC}{\mathcal{G}}

\newcommand{\ZC}{\mathcal{Z}}

\newcommand{\leak}{\text{leak}^{\text{EC}}_{\text{obs}}}
\begin{document}
% \singlespacing

\title{Security proof of practical quantum key distribution with detection-efficiency mismatch}

\author{Yanbao Zhang}
\affiliation{Institute for Quantum Computing, University of Waterloo, Waterloo, Ontario N2L 3G1, Canada}
\affiliation{Department of Physics and Astronomy, University of Waterloo, Waterloo, Ontario N2L 3G1, Canada}
 \affiliation{NTT Basic Research Laboratories and NTT Research Center for Theoretical Quantum Physics, NTT Corporation, 3-1 Morinosato-Wakamiya, Atsugi, Kanagawa 243-0198, Japan}

\author{Patrick J. Coles}
\affiliation{Institute for Quantum Computing, University of Waterloo, Waterloo, Ontario N2L 3G1, Canada}
\affiliation{Department of Physics and Astronomy, University of Waterloo, Waterloo, Ontario N2L 3G1, Canada}
\affiliation{Theoretical Division, Los Alamos National Laboratory, Los Alamos, NM 87545, US}

\author{Adam Winick}
\affiliation{Institute for Quantum Computing, University of Waterloo, Waterloo, Ontario N2L 3G1, Canada}

\author{Jie Lin}
\affiliation{Institute for Quantum Computing, University of Waterloo, Waterloo, Ontario N2L 3G1, Canada}
\affiliation{Department of Physics and Astronomy, University of Waterloo, Waterloo, Ontario N2L 3G1, Canada}

\author{Norbert L\" utkenhaus}
 \affiliation{Institute for Quantum Computing, University of Waterloo, Waterloo, Ontario N2L 3G1, Canada}
 \affiliation{Department of Physics and Astronomy, University of Waterloo, Waterloo, Ontario N2L 3G1, Canada}

% \date{\today}
\begin{abstract}
Quantum key distribution (QKD) protocols with threshold detectors are driving high-performance QKD demonstrations.  The corresponding security proofs usually assume that all physical detectors have the same detection efficiency. However, the efficiencies of the detectors used in practice might show a mismatch depending on the manufacturing 
and setup of these detectors. A mismatch can also be induced as the different spatial-temporal modes of an incoming signal might couple differently to a detector. Here we develop a method that allows to provide security proofs 
without the usual assumption. Our method can take the detection-efficiency mismatch into account without having to
restrict the attack strategy of the adversary. Especially, we do not rely on any photon-number cut-off of incoming
signals such that our security proof is directly applicable to practical situations.   
We illustrate our method for a receiver that is designed for polarization encoding and is sensitive to a number 
of spatial-temporal modes. In our detector model, the absence of quantum interference between any pair of  
spatial-temporal modes is assumed. For a QKD protocol with this detector model, we can perform a security proof 
with characterized efficiency mismatch and without photon-number cut-off assumption. Our method 
also shows that in the absence of efficiency mismatch in our detector model, the key rate increases if the loss 
due to detection inefficiency is assumed to be outside of the adversary's control, as compared to the view where for
a security proof this loss is attributed to the action of the adversary. 
\end{abstract}

\maketitle

\section{Introduction}
For practical quantum key distribution (QKD)~\cite{review:qkd2009} using photon-counting techniques (discrete variable QKD), information is usually encoded in optical signals that contain multiple photons. To decode the 
information, one measures the optical signals usually with threshold detectors which cannot tell apart
the number of incoming photons.  Security proofs of practical QKD protocols 
usually assume that all threshold detectors used have the same efficiency. 
Under this assumption, one can push the detection efficiency into the transmission channel, which is under 
the control of an adversary known as Eve. Thus the transmission loss and the inefficiencies of the detectors can be lumped together, and one can apply a security proof that applies to the  new increased effective transmission loss followed by ideal threshold detectors 
with perfect efficiency~\cite{lutkenhaus1999}.

In practice, however, it is not an easy job to build two detectors that have exactly the 
same efficiency. For example, the two detectors may be fabricated by different processes
and so a mismatch between their efficiencies is induced. In the presence of efficiency mismatch,
the different values for detection inefficiency cannot be lumped together and further treated as 
a single value for the loss over the transmission channel.  Therefore, existing security-proof 
techniques cannot be applied.

Even with a single detector, an efficiency mismatch can be induced by Eve.
Suppose that the response of this detector to a photon depends on its degrees 
of freedom such as spatial mode, frequency, or arrival time. These degrees of freedom 
are not necessarily being used to encode information. If Eve can manipulate these 
degrees of freedom, then an effective efficiency mismatch is induced. When the induced 
mismatch is large enough, powerful attacks on QKD systems exist, as demonstrated in 
Refs.~\cite{Zhao2008, Lydersen2010-2, Gerhardt2011, Rau2014, Shihan2015, Chaiwongkhot2019}. 
The intuition behind such attacks is as follows: The efficiency mismatch usually causes a specific 
outcome to be detected more frequently than the other outcomes in a chosen measurement 
basis; as a result, Eve can guess the outcomes correctly with a higher probability in 
the presence of efficiency mismatch than in the absence of efficiency mismatch. 
In typical experiments the efficiency mismatch may not seem significant, but it still 
means that the security cannot be formally proven by existing techniques.

In this paper we develop analytic tools that allow, subsequently, to prove with numerical 
methods the security in the presence of detection-efficiency mismatch. More precisely, 
we consider a set-up designed for polarization encoding, where each threshold detector used 
by the receiver Bob is sensitive to an incoming signal in a number of spatial-temporal modes. 
We assume a detector model where no quantum interference between any pair of spatial-temporal 
modes would take place as the incoming signal passes through the receiver or is being 
detected by a detector.  However, the optical loss experienced by the signal can depend 
on its spatial-temporal mode.  For the above detector model, the developed method can be 
applied given arbitrary characterized  efficiency mismatch.   To demonstrate our approach, 
we apply it to a Prepare{\&}Measure BB84-QKD protocol~\cite{BB84}.  Here we study the general 
case where the optical signals received by Bob may contain an unbounded number of photons such
that their states live in an infinite-dimensional space.  We can lower-bound 
the secret-key rate as a function of detection-efficiency mismatch and observed statistics. 
With our method, we can also study the individual effects of transmission loss and detection 
inefficiency on the secret-key rate. Our method is transferable to other QKD protocols.  
We note that Refs.~\cite{Fung2009, Lydersen2010, Bochkov2019, Ma2019} studied the security proof of 
the BB84-QKD protocol in the presence of efficiency mismatch \emph{but} under the assumption that Bob 
receives no more than one photon at each round. However, this assumption cannot be justified in practical
implementations of QKD where threshold detectors are being used. 

We also remark that the spatial-temporal-mode-dependent efficiency-mismatch models studied by 
us (see Sect.~\ref{sect:exp_con} for details) are different from those mode-dependent mismatch models 
studied in the previous work~\cite{Fung2009}. As we assume the absence of quantum interference between 
any pair of spatial-temporal modes throughout the measurement process, the measurement operators 
are block-diagonal with respect to various photon-number subspaces, where in 
each photon-number subspace the number of photons in each spatial-temporal mode is specified. See our 
previous work~\cite{Zhang2017} for the explicit expressions of these measurement operators. On the other 
hand, the previous work~\cite{Fung2009} studied the case that a quantum interference between a pair 
of auxiliary modes is possible, where the efficiency mismatch depends on these auxiliary modes.
Therefore, the detector model in Ref.~\cite{Fung2009} is more general than ours, albeit the security 
proof in Ref.~\cite{Fung2009} relies on a photon-number cut-off. Note that we believe that the 
interference between spatial-temporal modes will not play a significant role in a practical 
measurement setup. If we are wrong in this belief, our approach could be generalized 
to the more general detector model at the expense of more computational resources in our numerical 
key-rate evaluation. 

The rest of the paper is structured as follows: In Sect.~\ref{sect:exp_con} we describe the basic setup for an optical BB84-QKD implementation with a special emphasis on the description of the spatial-temporal modes coupled to the detectors. Then we explain our method in Sect.~\ref{sect:method}, where we also apply it to the described setup. In order to show the implication of our proof methods, we require a toy-model that describes what observations we would expect in real experiments, which we do in Sect.~\ref{sect:results}. There we also show the secret-key rates that we obtain for setups that exhibit  detection-efficiency mismatch. We summarize our findings in 
Sect.~\ref{sect:conclusion}. We note that all detectors considered in the rest of the paper are threshold detectors by
default.

\section{Experimental configuration} 
\label{sect:exp_con}
The method that we develop in this article is about the treatment and analysis of the detector. Therefore, to lay out and illustrate the method we develop,  it is sufficient to use the simple  BB84 protocol \cite{BB84}, which we consider with an ideal single-photon source, but with threshold detectors monitoring full optical modes. Without loss of generality, we use the polarization-encoding language. 

For our theoretical analysis, we use the entanglement-based formulation of Bennett, Brassard and 
Mermin~\cite{BBM1992}. This approach has been later generalized for general QKD protocols to the source-replacement 
scheme~\cite{Marcos2004}. This source-replacement scheme, in a thought-setup, realizes the source by preparing internally to the source a bi-partite entangled state. Measurements on one system effectively prepare the remaining system in the desired signal states with the prescribed probabilities. In the case of the BB84 protocol with an 
ideal single-photon source, the internal entangled state in the thought-set-up is the maximally entangled state
\begin{equation}
\Ket{\Phi}_{\textrm{AA}'}=\frac{1}{\sqrt{2}}(\Ket{H}_\textrm{A} \Ket{H}_{\textrm{A}'}+\Ket{V}_\textrm{A} \Ket{V}_{\textrm{A}'}), 
\label{eq:ent_AB}
\end{equation}
where $\Ket{H}$ and $\Ket{V}$ are horizontally and vertically polarized 
single-photon states, respectively. System $\textrm{A}'$  is prepared in the signal states of the BB84 protocol as  Alice uniformly randomly selects to measure the system $\textrm{A}$ in the horizontal/vertical ($H/V$) basis or the diagonal/anti-diagonal ($D/A$) basis.  System $\textrm{A}'$ enters the channel controlled by Eve and will emerge as system B at Bob's site. At that stage, the signal is not necessarily a single-photon signal, but can (due to Eve's action) be in any state of the optical modes supported by the detectors. For example, Eve might amplify the signal using an optical amplifier or replace the signals with multi-photon states at her discretion.  Bob thus has to perform a measurement on the full optical modes, not on the single-photon signals. In our setup, he randomly selects to measure the signal in either the $H/V$ basis or the $D/A$ basis of the optical modes supported by his device. We call the above procedure of preparing, distributing and measuring signal states a \emph{round}. 

After a large number of rounds, with the data recordings that detail Alice's effective signal choices and Bob's measurement outcomes, Alice and Bob continue the QKD protocol using the usual steps of testing, sifting, key map, error correction, and privacy amplification to obtain secret keys.  Our method can be easily generalized for other protocols that use, for example, weak coherent pulses as signal states, but the single-photon source example 
studied in this work is sufficient to demonstrate our method, which is about the detection side.

\begin{figure}% [htb!]
   \includegraphics[scale=0.58,viewport=5.5cm 14.5cm 16cm 20.5cm]{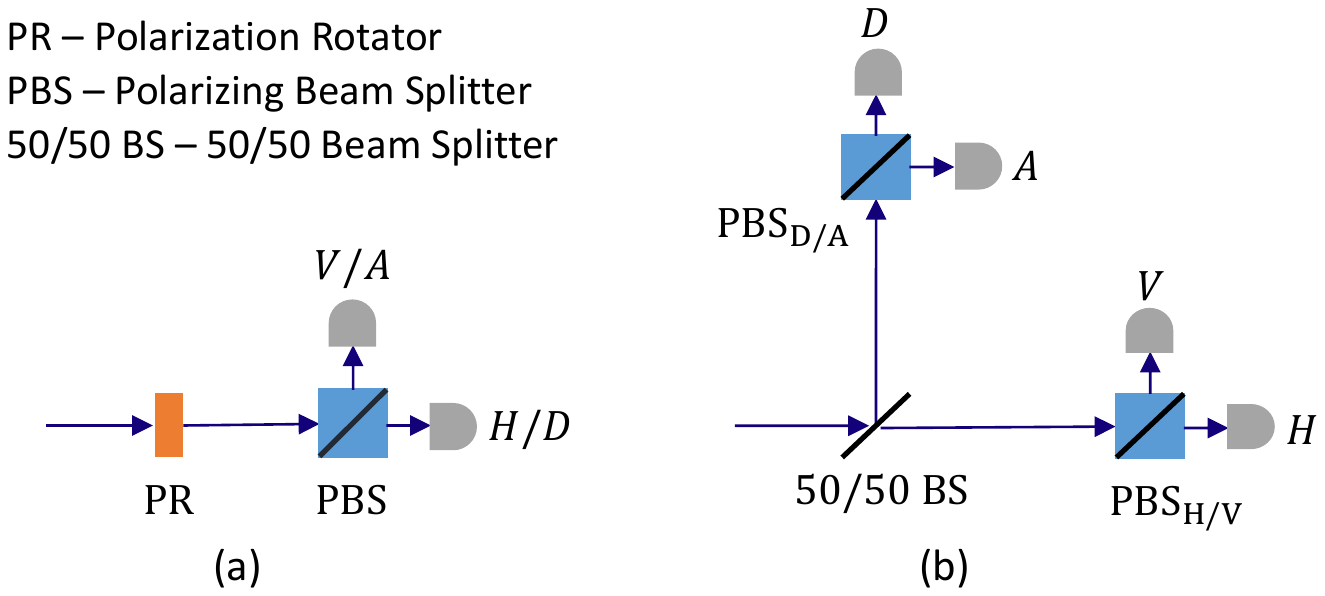} 
   \caption{Schematic of Bob's measurement device: panel (a) and (b) describe
   the active-detection and passive-detection schemes respectively. To 
   actively or passively select a measurement basis, a polarization rotator 
   or a 50/50 beam splitter is used. Under each basis, a polarizing beam splitter 
   and two detectors are used to measure the polarization state of an incoming 
   optical signal. Each detector is labelled by the associated measurement outcome.} 
   \label{setup} 
\end{figure}

So let us turn our attention to Bob's detection: Either the active- or passive-detection 
scheme, as depicted in Fig.~\ref{setup}, can be exploited.   As the detectors 
used in each scheme are threshold detectors, each detector can respond to an 
incoming optical signal only in two different ways, click or no click. The detectors might 
respond to different modes (frequency, timing, etc).

As stated in the introduction, there are two scenarios where an detection-efficiency mismatch may exist. 
Let us start with the first one. 
Due to the fabrications or setups in practice, the two detectors shown in Fig.~\ref{setup}(a) for the 
active-detection scheme may have different efficiencies $\eta_{H/D}$ and $\eta_{V/A}$. Similarly, 
the four detectors in Fig.~\ref{setup}(b) for the passive-detection scheme may have efficiencies 
$\eta_H$, $\eta_V$, $\eta_D$ and $\eta_A$ respectively. Here, the subscripts indicate the detectors 
used in a scheme. We call this kind of mismatch the spatial-temporal-mode-independent mismatch, 
in contrast to the following mismatch which depends additionally on the spatial-temporal modes 
chosen by Eve.   

The second scenario is that of an active Eve.  By manipulating the spatial-temporal 
mode of an optical signal, Eve can change the coupling of the signal with a detector, 
resulting in a change in the effective detection efficiency of the detector. 
Especially in free-space QKD it is possible for Eve to change the angle of an incoming  
signal~\cite{Rau2014,Shihan2015, Chaiwongkhot2019} to influence the coupling of the signal with the active 
detection area of a detector, while for fiber-based signals simple time delays can be 
introduced~\cite{Zhao2008} to exploit uneven aligned detection time windows. 
Therefore, in a setup with several detectors,  the efficiencies of these detectors can not only 
differ from each other but also depend on the spatial-temporal modes coupled to the detectors,
giving rise to the so-called spatial-temporal-mode-dependent mismatch. In this work 
we analyze the security in both above scenarios. 
  
Bob's detectors may respond to a large number of spatial-temporal modes. If the detection 
efficiencies related to these modes differ strongly from each other, it might become possible 
for Eve to control Bob's detection events thoroughly by sending the signals to the modes that 
couple particularly well only to a specific detector of Bob 
for which Eve desired to cause a detection event. For this attack to be possible in its extreme 
form, the number of modes must be equal to, or larger than, the number of detectors in the setup. 
For this reason, we choose the number of controllable modes to be equal to the number of detectors. 
In order to obtain visually simple illustrations of the secret-key rates, we choose mismatch models 
parametrized by two values for the efficiencies: a high value $\eta_1$ for one 
detector, and a lower value $\eta_2$ for the other detectors, as shown in Tables~\ref{active_mismatch} 
and~\ref{passive_mismatch}.  We emphasize that these mismatch models are considered just for 
ease of visual presentation, as the approach developed here can be exploited with an arbitrary 
mismatch model. To analyze the security of QKD systems, for example in a certification process, 
the choice of the mismatch model and its parameters will need to be justified and characterized  
in practice.

\begin{table}
\caption{Spatial-temporal-mode-dependent mismatch model in the active-detection scheme, 
where $0\leq \eta_2 \leq \eta_1 \leq 1$. The efficiencies of the two detectors labelled in 
Fig.~\ref{setup}(a) are listed in a column, where each column corresponds to 
a spatial-temporal mode.}
\begin{center} 
\begin{tabular}  {|c |c|c|}

 \hline
 & Mode 1  &   Mode 2  \\ 
  \hline
Detector `$H/D$'  & $\eta_1$ & $\eta_2$  \\
  \hline
Detector `$V/A$' & $\eta_2$ & $\eta_1$   \\
  \hline
\end{tabular}
\newline \\
\end{center}
\label{active_mismatch} 
\end{table}

\begin{table}
\caption{Spatial-temporal-mode-dependent mismatch model in the passive-detection scheme, 
where $0\leq \eta_2 \leq \eta_1 \leq 1$. The efficiencies of the four detectors labelled in 
Fig.~\ref{setup}(b) are listed in a column, where each column corresponds to 
a spatial-temporal mode.}
\begin{center}
\begin{tabular}  {|c |c |c| c| c|}
 \hline 
 & Mode 1 & Mode 2 & Mode 3 & Mode 4 \\
  \hline
Detector `$H$'  & $\eta_1$ & $\eta_2$ & $\eta_2$ & $\eta_2$  \\
  \hline
Detector `$V$'  & $\eta_2$ & $\eta_1$  & $\eta_2$ & $\eta_2$ \\
  \hline
Detector `$D$'  & $\eta_2$ & $\eta_2$ & $\eta_1$ & $\eta_2$ \\
  \hline
Detector `$A$'  & $\eta_2$ & $\eta_2$ & $\eta_2$ & $\eta_1$ \\
  \hline
\end{tabular}
\newline \\
\end{center}
\label{passive_mismatch} 
\end{table}

\section{Key-rate calculation method} 
\label{sect:method}

\subsection{Formulation of key-rate calculation as a convex-optimization problem}
\label{sect:convex_optimization}
The asymptotic key rate certifiable against all collective attacks~\cite{Devetak2005}  
is given by the difference between two terms, which are associated with privacy amplification 
(PA) and error correction (EC) respectively. The EC term depends only the the measurement 
statistics and can be calculated without any further information on the implementation 
of the QKD protocol.  The main difficulty 
of the security proof relies on how to obtain a lower bound on the PA term. As shown in 
Refs.~\cite{Coles2016,Winick2018}, a reliable numerical lower bound 
on the PA term can be provided by solving a convex-optimization problem. In the following, 
we will give a brief review of the theory behind that reformulation. 

In a generic QKD protocol, the measurement statistics in an experiment are summarized as a probability 
distribution $p_{\textrm{AB}}(x,y)$, where $x$ and $y$ are random variables corresponding to 
the events detected by  Alice and Bob respectively. The corresponding measurement operators are $M_x^{\textrm{A}}$ 
and $M_y^{\textrm{B}}$. In addition, for the techniques shown in this paper, we will be able to 
provide from experimental observations lower bounds on the probability of at most $k$ photons arriving at Bob. 
These bounds will be brought in as additional explicit constraints in the convex-optimization problem. 
To formulate the corresponding constraints, we introduce the projectors $\Pi_k$ onto 
the photon-number subspace of Bob containing at most $k$ photons, and the corresponding lower bound on its 
expectation value as $b_k$. Then, the calculation of the PA term, denoted by $\alpha$, can be written as the 
convex-optimization problem~\cite{Coles2016,Winick2018}
\begin{equation}\label{eq:key_minimization}
\begin{array}{rc}
\alpha:=& \min_{\rho_{\textrm{AB}}}  D(\GC(\rho_{\textrm{AB}}) || \ZC(\GC(\rho_{\textrm{AB}}))) \\
 \text{subject to}& \rho_{\textrm{AB}}\geq 0, \Tr\left(\rho_{\textrm{AB}}\right)=1  \\
& \Tr\big((M_x^{\textrm{A}}\otimes M_y^{\textrm{B}})\rho_{\textrm{AB}}\big) =p_{\textrm{AB}}(x,y)\\
& \Tr(\Pi_k \rho_{\textrm{AB}})\geq b_k\; .
\end{array}
\end{equation}
Here, $D(\sigma ||\tau):= \Tr(\sigma \log_2 \sigma) -\Tr(\sigma \log_2 \tau)$ is the relative entropy, $\GC$ is 
the post-selection map, and $\ZC$ is the quantum channel describing the key map of the QKD protocol (see below for the details). In our applications we will later choose for $k\in\{1,2\}$, or use even the constraints for both values of $k$.  We remark that both the objective function and constraints are convex in the optimization variable $\rho_{\textrm{AB}}$. 

Once we obtain a reliable lower bound $\beta$ on the PA term $\alpha$ of Eq.~\eqref{eq:key_minimization} as 
$\beta \leq \alpha$ according to the numerical method developed in Ref.~\cite{Winick2018}, the asymptotic key rate 
$K$ per round is bounded by
\begin{equation}\label{eq:key_lower_bound}
K \geq K_\text{lb}\doteq \beta -\leak,
\end{equation}
where $\leak$ denotes the amount of information leaked to Eve per round of the protocol 
during error correction. This takes automatically into account any post-selection mechanism of the protocol,
 as any jointly discarded signals do not cause an error-correction cost. 
Likewise, the PA cost $\beta$ automatically takes care of the same post-selection process, so that the total key rate $K$ is counted as per round of the protocol. As we are discussing key rates in the asymptotic limit of a large number of exchanged signals, the reduction by any fraction of signals that is utilized to estimate the observed probability distribution $p_{AB}(x,y)$ of measurement results and other finite-size effects are negligible. Furthermore, the 
security proofs under collective and coherent attacks are equivalent in this limit~\cite{renner:qc2005}, and hence 
our key-rate lower bound $K_\text{lb}$ in principle holds for coherent attacks.   We remark 
that the numerical method developed in Ref.~\cite{Winick2018} obtains a key-rate lower bound by the following two 
steps: First, by an iterative method we find a near-optimal solution of the convex-optimization problem in 
Eq.~\eqref{eq:key_minimization} and thus an upper bound on the PA term $\alpha$; second, we take advantage of the 
duality principle satisfied by convex optimization to obtain a reliable lower bound $\beta$ on the PA term 
$\alpha$. The key-rate lower bound $K_\text{lb}$ obtained according to the numerical method developed in 
Ref.~\cite{Winick2018} is reliable in the sense that the lower bound $K_\text{lb}$ is valid even considering 
the finite precision in floating-point representations. Moreover,  % as noted in Ref.~\cite{Winick2018}, 
the imprecision in function evaluations is estimated to be at the level of $10^{-8}$ according to the 
CVX package used by us for solving convex programs~\cite{cvx,gb08}, although a rigorous analysis is currently 
missing and deserves further investigation in future work. The estimated function-evaluation imprecision is much 
smaller than the numerical key-rate lower bounds reported in Sect.~\ref{sect:results}, suggesting that our 
numerical key-rate lower bounds are reliable even considering the effect of function-evaluation imprecision. 
% At the same time, we can be bounded in principle~\cite{Winick2018}, 
 % implementation-independent errors such as floating-point representation errors. 

The map $\GC$ in the objective function of Eq.~\eqref{eq:key_minimization} describes the post-selection after Alice's
and Bob's public announcements for sifting. For simplicity, we concentrate here on the case where to distill secret keys Alice and Bob keep only those signals where both measured in the $H/V$ basis.  We also note that for optical implementations, the announcements usually used for sifting are slightly more involved than the simple basis-dependent sifting of the BB84 protocol. The reason is that the potential presence of multiple photons in the incoming signals can cause several detectors to show detection events simultaneously. If Bob uses the active-detection scheme, the sifting announcement by Bob consists of the declaration whether he used the  $H/V$ basis measurement, and whether at least one detector fired.  However, if Bob uses the passive-detection scheme, we have to decide what to do with events where we have multiple detections across the groups associated with different polarization bases (cross clicks), for example both the $H$ and the $D$ detector firings. Here we make the choice to keep only those events where either the $H$, the $V$, or both the $H$ and $V$ (denoted as $HV$ event)  detectors fire, while all other events (no clicks, clicks only in any of the $D$ and the $A$ detectors, or cross clicks) are being discarded. In order to achieve this goal, Alice publicly announces the basis choice where one of two bases is chosen uniformly randomly at each round, and  Bob announces whether the desired events are observed. This corresponds to applying the post-selection map 
\begin{equation}\label{eq:post-selection_map}
\GC (\rho_{\textrm{AB}})=G \rho_{\textrm{AB}} G^{\dagger},
\end{equation}  
where $G=\frac{1}{\sqrt{2}}\one^{\textrm{A}} \otimes \sqrt{M_H^{\textrm{B}}+M_V^{\textrm{B}}+M_{HV}^{\textrm{B}}}$ is a Kraus operator.  Here, $\one^{\textrm{A}}$ is the identity operator in the state space of Alice, and the 
positive-operator valued measure (POVM) elements $M_H^{\textrm{B}}$, $M_V^{\textrm{B}}$, and $M_{HV}^{\textrm{B}}$ for Bob have been derived in Appendix A and B of Ref.~\cite{Zhang2017}, with the remark that for the active-detection scheme we need to put the coefficient $1/2$ before each POVM element shown in Ref.~\cite{Zhang2017} to account for  Bob's probability of selecting each measurement basis.

After the public announcements and the corresponding post-selection step, Alice chooses a key map, which is represented by a quantum channel $\ZC$. The key map is a function whose input is Alice's measurement outcome in the key-generation basis and whose output is a key value, $0$ or $1$. Suppose that we make a particular choice of key map here, namely that Alice's outcomes $H$ and $V$ are mapped to key values $0$ and $1$, respectively, and that the corresponding POVM elements $M_H^{\textrm{A}}$ and $M_V^{\textrm{A}}$ are projective (see Appendix~\ref{sect:Alice_operators}). The application of the key map corresponds to the application of the quantum channel 
\begin{align}\label{eq:key_map}
\ZC(\GC(\rho_{\textrm{AB}}))&=(M_H^{\textrm{A}}\otimes \one^{\textrm{B}}) \GC(\rho_{\textrm{AB}}) (M_H^{\textrm{A}}\otimes \one^{\textrm{B}}) \notag \\
&+(M_V^{\textrm{A}}\otimes \one^{\textrm{B}}) \GC(\rho_{\textrm{AB}}) (M_V^{\textrm{A}}\otimes \one^{\textrm{B}}). 
\end{align}

Given the measurement statistics $p_{\textrm{AB}}(x,y)$, the lower bounds $b_k$ on the photon-number distribution, 
the post-selection map $\GC$, and the key-map-realizing quantum channel $\ZC$, 
in principle we can run numerical optimization to obtain a reliable lower bound of the minimization problem in 
Eq.~\eqref{eq:key_minimization}. However, 
for the situation studied, the number of photons arriving at Bob is unbounded and so the dimension 
of the quantum state $\rho_{\textrm{AB}}$ is infinite. For this reason we need to develop techniques that allow us 
to simplify the optimization problem such that a reliable key-rate lower bound can be numerically obtained  
by optimizing over only finite-dimensional quantum states. These techniques are described in the next two subsections.

\subsection{Simplification of the convex-optimization problem: Flag-state squasher}
\label{sect:simplification}
Since Bob's measurement POVMs are block-diagonal with respect to 
the subspaces associated with total photon numbers across all modes~\cite{Zhang2017}, 
we can assume without loss of generality that Eve performs a quantum non-demolition (QND) 
measurement of the total photon number after her interaction with the signals, 
and before  their arrivals at Bob's side. As a consequence, the state $\rho_{\textrm{AB}}$ can be assumed, 
without loss of generality, to be block-diagonal with the same subspace structure, meaning 
that the state takes the form
\begin{equation}
\rho_{\textrm{AB}}=\bigoplus_{n=0}^{\infty}p_n\rho_{\textrm{AB}}^{(n)}.
\label{rho_block-diagonal}
\end{equation}
The weight of each subspace carrying a total number of $n$ photons is given by the 
corresponding probability $p_n$, and the corresponding normalized conditional state 
is denoted by  $\rho_{\textrm{AB}}^{(n)}$.
% In order to obtain 
%a key-rate lower bound valid under this attack, we can optimize over only the states 
%$\rho_{\textrm{AB}}$ that are block-diagonal with respect to various photon-number subspaces.
%We can also justify the assumption by the following two observations: 1) The objective function 
%in Eq.~\eqref{eq:key_minimization} is a coherence measure~\cite{Coles2016, Coles2012}; 2) 
%The QND measurement corresponds to an incoherent completely positive and trace preserving (CPTP) 
%map, and the coherence measure cannot increase under any incoherent CPTP map~\cite{Baumgratz2014}. 

Considering the block-diagonal structure of the state and  Bob's measurement POVMs, we can write 
\begin{align} \label{eq:actual_case}
& \rho_{\textrm{AB}}=p_{n\leq k} \rho_{\textrm{AB}}^{(n\leq k)} \; \bigoplus \;  (1-p_{n\leq k})\rho_{\textrm{AB}}^{(n>k)},\notag \\
& M_y^{\textrm{B}}=M_{y,n\leq k}^{\textrm{B}} \; \bigoplus \; M_{y,n>k}^{\textrm{B}}, 
\end{align}
where $k$ is a free parameter chosen in the security proof and $p_{n\leq k}$ is the probability 
that no more than $k$ photons arrive at Bob. The $(n \leq k)$-photon subspace is of finite
dimension, which is compatible with the numerical key-rate optimization framework.  
On the other hand, the $(n>k)$-photon subspace is infinite dimensional, which is not directly suitable to be handled  
by our numerical method. To resolve this problem, we introduce the {\em flag-state squasher}. The general framework of squashing models that map large-dimensional state/measurement descriptions without loss of generality to 
lower-dimensional systems has been described in Refs.~\cite{normand2008, Kiyoshi2008, Gittsovich2014}.

\begin{theorem}{\bf Flag-State Squasher}
\label{th:flagstatesquasher}
Suppose that we have a POVM with elements $M_y$, where $y \in \{1, \dots, J\}$, such that each element can be written in a block-diagonal form $M_{y,n\leq k}\oplus M_{y,n>k}$, with an associated Hilbert space structure given by ${\cal H}_{n \leq k} \oplus {\cal H}_{n>k}$.  Then there exists a completely positive trace preserving (CPTP) map $\Lambda$  (referred to as a squashing map) from ${\cal H}_{n \leq k} \oplus {\cal H}_{n>k}$ to ${\cal H}_{n\leq k} \oplus {\cal H}_{J}$, where the dimension  $dim({\cal H}_{J})=J$, such that $\Tr\left( \rho M_y \right) = \Tr\big(\Lambda(\rho) \tilde{M}_y \big) \; \;  \forall \rho \in {\cal H}_{n\leq k} \oplus {\cal H}_{n>k}$  with 
\begin{equation}
\label{eq:virtual_case}
\tilde{M}_y = M_{y,n\leq k} \oplus |y\rangle\langle y| \; ,
\end{equation}
 where the states $|y\rangle$ form an orthonormal basis of ${\cal H}_{J}$. 
\end{theorem}

\begin{proof}
 We need to show that the CPTP map $\Lambda$ exists with the desired properties. This can be done by an explicit construction as indicated in Fig.~\ref{flagsquasher}. For this purpose, we consider a general input state given in block form $\rho = \left( \begin{array}{cc} \rho_{\textrm{ss}} & \rho_{\textrm{sl}}\\ \rho_{\textrm{ls}}& \rho_{\textrm{ll}}\end{array} \right)\; ,$  where index `s' refers to the small subspace ${\cal H}_{n\leq k}$ and index `l' to the large subspace ${\cal H}_{n>k}$.  We can then describe the action of the squashing map $\Lambda$ by its action onto an arbitrary input state of the above form as
 \begin{equation}
\label{eq:virtual_case_state}
\Lambda(\rho)=  \left( \begin{array}{cc} \rho_{\textrm{ss}} & 0\\ 0 & \sum_{y=1}^J \Big(\Tr\left(\rho_{\textrm{ll}} M_{y,n>k} \right) \; |y\rangle\langle y|\Big) \end{array} \right)\; .
\end{equation}
It is straightforward to see that the state $\Lambda(\rho)$ satisfies the properties 
$\Tr\left(\rho M_y \right) = \Tr\big(\Lambda(\rho) \tilde{M}_y \big) \; \; \forall y$ 
required by a flag-state squasher. 
\end{proof}

\begin{figure}%[htb!]
   \includegraphics[scale=0.45, viewport=7.5cm 8cm 14.5cm 12.5cm]{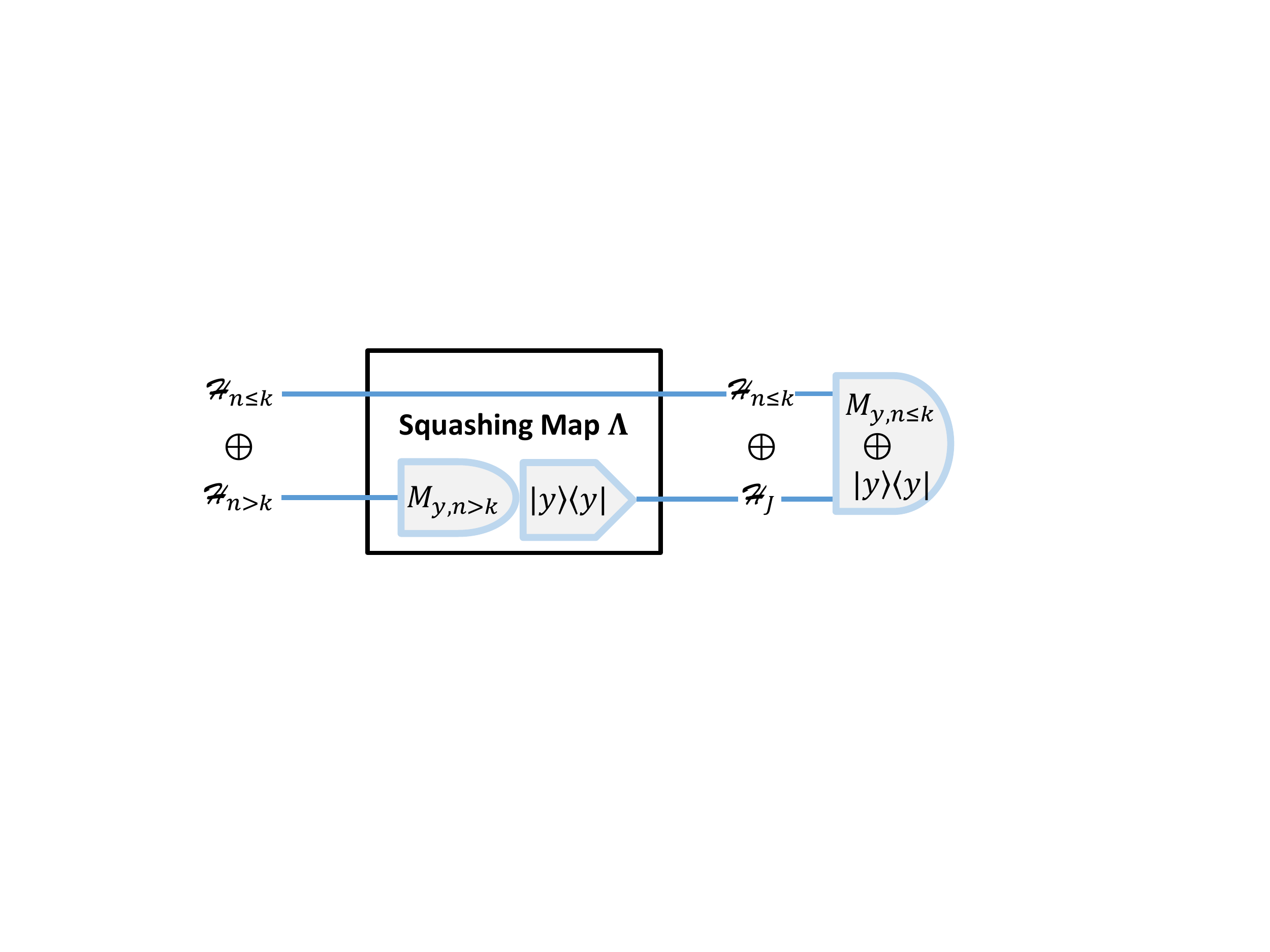} 
   \caption{Constructive description of the squashing map $\Lambda$ for the flag-state squasher. Each line corresponds to a subspace of the input Hilbert space associated with the block-diagonal decomposition of the POVM elements $M_y$ with 
   $y \in \{1, \dots, J\}$, as indicated on the left side.} 
   \label{flagsquasher} 
\end{figure}

The large subspace ${\cal H}_{n>k}$, which in the case of Bob's measurement is infinite dimensional, is simply reduced to a smaller subspace ${\cal H}_{J}$ by performing the measurement 
$\{M_{y,n>k}, y=1,\dots, J\}$ on ${\cal H}_{n>k}$ and flagging the result $y$ into an orthogonal register which replaces the original large subspace. This approach of creating squashing models to smaller Hilbert spaces relies only on the block-diagonal structure 
of the original POVM elements. As soon as that assumption is met, a flag-state squasher can be constructed. In this
 work we apply Theorem~\ref{th:flagstatesquasher} to Bob's states and measurements. As Bob's measurements are 
 block-diagonal for an arbitrary choice of $k$ (see Eq.~\eqref{eq:actual_case}), we can freely choose which large 
 photon-number subspace of Bob to be flagged. For this reason, we refer to the free parameter $k$ as the 
 \emph{photon-number flag threshold}.

As in any case where a squashing map exists mapping the original measurement to an alternative measurement of    
 a smaller dimension, we can assume that the squashing map is part of Eve's action. As a result, we overestimate Eve's power (see below for a detailed explanation), but as a trade-off we can now assume without further loss of generality that Bob receives signals in a reduced, finite-dimensional Hilbert space. So the key-rate optimization problem in 
 Eq.~\eqref{eq:key_minimization} formulated with the squashed states of the form in 
 Eq.~\eqref{eq:virtual_case_state} and POVM elements of the form in Eq.~\eqref{eq:virtual_case}, which is a finite-dimensional convex-optimization problem, will provide a lower bound on the secret-key rate in the actual implementation. Note, however, that the virtual POVM element components  $\tilde{M}_{y,n>k}^{\textrm{B}}=|y\rangle\langle y|$ are projective and orthogonal. Therefore, when her QND measurement result of the total photon number is $n>k$ 
 Eve could perform a strong attack by measuring the incoming signals from Alice with Bob's actual measurement 
 $\{M_{y, n>k}^{\textrm{B}}: y=1,\ldots, J\}$ and then preparing/sending to Bob the flag state $|y'\rangle$  corresponding to her measurement result $y'$. This attack would deterministically trigger the same result $y'$ when Bob performs the virtual measurement $\{\tilde{M}_{y, n>k}^{\textrm{B}}: y=1,\ldots, J\}$ according to the squashing map. Hence, by attributing the squashing map to Eve's action, Eve could completely learn every result of Bob when $n>k$, and so we overestimate the power of Eve as compared with in the actual implementation. For this reason, the flag-state squasher must be accompanied by a constraint that limits the resulting state mostly to the $(n\leq k)$-photon subspace, which is given by the bound $b_k$ in our optimization problem of Eq.~\eqref{eq:key_minimization}.

Finally, we remark that without loss of generality the states $\rho_{\textrm{AB}}$ and $\rho_{\textrm{AB}}^{(n)}$ 
can be assumed to be real-valued. This is because all measurement POVM elements $M_x^{\textrm{A}}$ and $M_y^{\textrm{B}}$ of Alice and Bob can be represented by real-valued matrices and because the objective function 
to be minimized for bounding the key rate in Eq.~\eqref{eq:key_minimization} is a convex function of the state $\rho_{\textrm{AB}}$. For detailed proofs, see for example Sect. V C in Ref.~\cite{ferenczi2012a}.  
We also emphasize that the block-diagonal structure and the real-matrix representation of the state 
$\rho_{\textrm{AB}}$ apply to both the active- and passive-detection schemes. By using a real-matrix representation
of $\rho_{\textrm{AB}}$, the number of free parameters in the key-rate optimization problem of  
Eq.~\eqref{eq:key_minimization} is reduced.

\subsection{Constraints on photon-number distribution}
\label{sect:photon_number_bounds}
To solve the convex-optimization problem in Eq.~\eqref{eq:key_minimization}, we need make use of a flag-state squasher as introduced in Theorem \ref{th:flagstatesquasher} where the small subspace will be chosen to be the incoming subspace containing at most $n=1$ photon, or at most $n=2$ photons. In order to obtain positive key rates, it will be necessary to show that the overlap of the incoming states with this subspace can be lower-bounded by some number $b_k$, $k=1$ or $2$. Following the numerical method developed in our previous study of entanglement verification with efficiency mismatch~\cite{Zhang2017}, we obtain such bounds directly from 
the experimentally observed measurement statistics $p_{\textrm{AB}}(x,y)$. The intuition behind this approach is that higher photon numbers will necessarily lead to double clicks, cross clicks, and/or errors. 

This way of using experimental observations to bound the photon-number distribution  was first established in 
Ref.~\cite{lutkenhaus1999} and further refined in Ref.~\cite{koashi2008}, and then  extended to 
the case of inefficient detectors in Ref.~\cite{Zhang2017}. Note that the theoretical approach is 
independent of the number of spatial-temporal modes that we use (in addition to the polarization degree of freedom). We demonstrate the results of our method here for the two-mode case (with the active-detection scheme) and for the four-mode case (with the passive-detection scheme).

Before explaining the method, we would like to point out that the two properties of the state $\rho_{\textrm{AB}}$
discussed in the above subsection, i.e., the block-diagonal structure with respect to 
various photon-number subspaces and the real-number representation of the density matrix, 
will be used also in the optimization problems formulated in this subsection.   
The second property helps to reduce the number of free parameters in the optimization.

\subsubsection{Active-detection case}
As stated before, the intuition is that as an increasing number of photons are received 
by Bob, the probability of double clicks (clicks at both detectors) will increase and finally surpass the 
double-click probability observed in an experiment. Similar arguments hold for an effective error, which we define below. Thus we will show that the experimental observations allow us to put an upper bound on the probability that the signals received by Bob contain more than any given number of photons.

In order to make this intuition precise, we start by defining the double-click operator
\begin{equation}
F_{DC}=\frac{1}{2} \one^{\textrm{A}}\otimes M_{HV}^\textrm{B}+\frac{1}{2} \one^{\textrm{A}}\otimes M_{DA}^\textrm{B}, 
\label{eq:active_dc_op}
\end{equation}
and the effective-error operator
\begin{align}
F_{EE}&= \frac{1}{2} M_{H}^{\textrm{A}}\otimes (M_{V}^{\textrm{B}}+\frac{1}{2}M_{HV}^{\textrm{B}})+\frac{1}{2} M_{V}^{\textrm{A}}\otimes (M_{H}^{\textrm{B}}+\frac{1}{2}M_{HV}^{\textrm{B}})\notag \\
& +\frac{1}{2} M_{D}^{\textrm{A}}\otimes (M_{A}^{\textrm{B}}+\frac{1}{2}M_{DA}^{\textrm{B}})+\frac{1}{2} M_{A}^{\textrm{A}}\otimes (M_{D}^{\textrm{B}}+\frac{1}{2}M_{DA}^{\textrm{B}}),
\label{eq:active_error_op}
\end{align}
where the pre-factor $1/2$ at each term describes the probability to choose the corresponding measurement basis. 
The form of the effective-error operator is chosen according to the squashing model~\cite{normand2008, Kiyoshi2008} for the active-detection scheme: the double-click events are mapped uniformly randomly in a post-processing step to either of the two single-click events associated with the chosen basis.  In Eqs.~\eqref{eq:active_dc_op} 
and~\eqref{eq:active_error_op}, Alice's measurement operators are ideal measurement operators in the one-photon space
(see Appendix~\ref{sect:Alice_operators}), while Bob's measurement operators are described
in Appendix A and B of Ref.~\cite{Zhang2017}.

We formalize the above intuition by studying the following optimization problems
\begin{equation}
\label{eq:min_dc}
\begin{array}{rc}
d_{n,min}=& \min_{\rho_{\textrm{AB}}^{(n)}} \Tr\left(\rho_{\textrm{AB}}^{(n)}F_{DC}^{(n)}\right) \\
% \left(d_n\equiv \Tr\left(\rho_{\textrm{AB}}^{(n)}F_{DC}^{(n)}\right)\right) \\
 \text{subject to}&  \rho_{\textrm{AB}}^{(n)}\geq 0   \\
& \Tr\left(\rho_{\textrm{AB}}^{(n)}\right)=1  
\end{array}
\end{equation}
and 
\begin{equation}
\label{eq:min_error}
\begin{array}{rc}
e_{n,min}=& \min_{\rho_{\textrm{AB}}^{(n)}} \Tr\left(\rho_{\textrm{AB}}^{(n)}F_{EE}^{(n)}\right) \\
% \left(e_n\equiv\Tr\left(\rho_{\textrm{AB}}^{(n)}F_{EE}^{(n)}\right)\right) \\
 \text{subject to}&  \rho_{\textrm{AB}}^{(n)}\geq 0   \\
& \Tr\left(\rho_{\textrm{AB}}^{(n)}\right)=1  \; .
\end{array}
\end{equation}
The operators $F_{DC}^{(n)}$ and $F_{EE}^{(n)}$ are projections of the operators $F_{DC}$
and $F_{EE}$ onto the $n$-photon subspace of Bob. We remark that the above 
optimizations are over all possible $n$-photon states $\rho_{\textrm{AB}}^{(n)}$, 
while the optimization problems formulated in our previous study of entanglement 
verification with efficiency mismatch~\cite{Zhang2017} run over only the 
states $\rho_{\textrm{AB}}^{(n)}$ satisfying the positive-partial-transpose 
criterion~\cite{Peres1996,Horodecki1996}.

The optimization problems described by Eqs.~\eqref{eq:min_dc} and~\eqref{eq:min_error} 
have the form of semi-definite programs (SDPs). In order to solve them, we utilize 
the YALMIP~\cite{YALMIP} toolbox in MATLAB. From our calculations we make the observation 
that the minimum double-click probability $d_{n,min}$ is monotonically 
non-decreasing as the photon number $n$ goes up. We therefore obtain the inequality
\begin{equation} \label{eq:doubleclick_rela}
d_{n,min} \geq d_{3,min}, \forall  n\geq 3.
\end{equation}
Moreover, we observed the inequality relations 
\begin{equation}\label{eq:efferror_rela1}
e_{n,min} \geq e_{3,min}, \forall  n\geq 3,
\end{equation}
and  
\begin{equation}\label{eq:efferror_rela2}
e_{n,min} \geq e_{min} := \min\{ e_{2,min}, e_{3,min} \}, \forall  n\geq 2.
\end{equation}
We would like to point out that we did not go through the effort to prove the above 
inequalities with analytical methods, though the numerical evidence strongly supports that 
these inequalities hold for an arbitrary active-detection efficiency mismatch. 
In Figs.~\ref{active_doubleclick} and~\ref{active_error}, we report our numerical evidence 
for the specific mismatch model of Table~\ref{active_mismatch}. Especially, one can see from 
these figures that the curve becomes monotonous as the efficiency mismatch increases. 

\begin{figure}[htb!]
\includegraphics[scale=0.53,viewport=3cm 8.6cm 19cm 19.5cm]{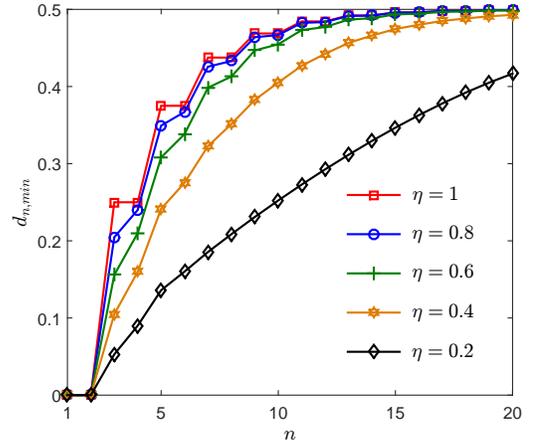}  
\caption{The minimum double-click probability $d_{n,min}$ vs. the photon number 
 $n$ received by Bob for the active-detection mismatch model of 
 Table~\ref{active_mismatch} with $\eta_1=1$ and $\eta_2=\eta$.  Note the monotonicity 
 of each curve as a function of $n$ and that $d_{2,min}$, as well as $d_{1,min}$, is 
 always equal to zero.}
\label{active_doubleclick} 
\end{figure}

In view of Eqs.~\eqref{eq:doubleclick_rela} and~\eqref{eq:efferror_rela1},  we find 
that the double-click probability $d_\text{obs}$ and the effective-error probability
$e_\text{obs}$ observed in practice satisfy 
\begin{equation}
d_\text{obs}=\sum_{n=0}^{\infty}p_n \Tr\left(\rho_{\textrm{AB}}^{(n)}F_{DC}^{(n)}\right) \geq (1-p_0-p_1-p_2)d_{3,min}, \label{eq:actual_doubleclick_bound} 
\end{equation}
and
\begin{equation}
e_\text{obs}=\sum_{n=0}^{\infty}p_n\Tr\left(\rho_{\textrm{AB}}^{(n)}F_{EE}^{(n)}\right) \geq (1-p_0-p_1-p_2)e_{3,min},\label{eq:actual_error_bound}
\end{equation}
by using that $\sum_{n=0}^{\infty}p_n=1$. Hence, we can set the bound  $b_2 \leq p_0+p_1+p_2$ as
\begin{equation}\label{eq:active_b2_bound}
 b_2=1-\min\big(\frac{d_\text{obs}}{d_{3,min}},\frac{e_\text{obs}}{e_{3,min}}\big).
\end{equation}
Note that for the observations simulated in Sect.~\ref{sect:results}, we found that 
$\frac{d_\text{obs}}{d_{3,min}}<\frac{e_\text{obs}}{e_{3,min}}$ and therefore the bound
$b_2=1-\frac{d_\text{obs}}{d_{3,min}}$. Similarly, in view of Eq.~\eqref{eq:efferror_rela2} we have 
\begin{equation}
e_\text{obs}=\sum_{n=0}^{\infty}p_n\Tr\left(\rho_{\textrm{AB}}^{(n)}F_{EE}^{(n)}\right) \geq (1-p_0-p_1)e_{min}. 
\end{equation}
Thus we can obtain a bound $b_1 \leq p_0+p_1$ as % $(p_0+p_1)$ from below. 
\begin{equation}\label{eq:active_b1_bound}
b_1=1-\frac{e_\text{obs}}{e_{min}}.
\end{equation}
In this case, the double-click estimations do not lead to a non-trivial bound on $b_1$ as there exist 
two-photon states that do not lead to double clicks ($d_{2,min}=0$), see Fig.~\ref{active_doubleclick}. 

\begin{figure}[htb!]
\includegraphics[scale=0.53,viewport=3cm 9cm 19cm 19.5cm]{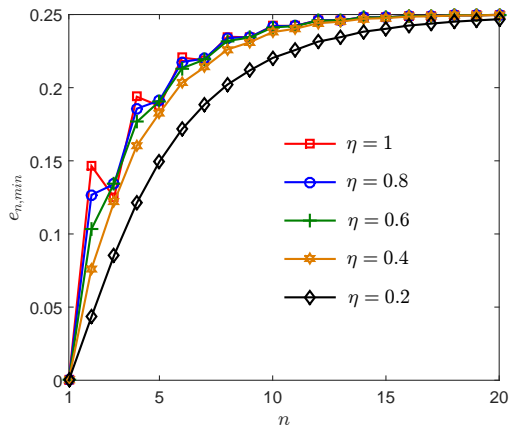}  
\caption{The minimum effective-error probability $e_{n,min}$ vs. the photon number $n$
received by Bob for the active-detection mismatch model of Table~\ref{active_mismatch} 
with $\eta_1=1$ and $\eta_2=\eta$.  Note that $e_{3,min}$ is a lower bound on $e_{n,min}$ 
when $ n \geq 3$ and that $e_{1,min}$ is always equal to zero.}\label{active_error} 
\end{figure}

The above bounds $b_1$ and $b_2$ % in Eqs.~\eqref{eq:active_b2_bound} and~\eqref{eq:active_b1_bound}
together with the flag-state squasher approach for the corresponding subspaces
can be used in the key-rate optimization problem of Eq.~\eqref{eq:key_minimization} when the active-detection scheme is used.

\subsubsection{Passive-detection case}
The passive-detection scheme utilizes a 50/50 beam splitter to passively select a measurement basis, as shown in 
Fig.~\ref{setup}(b). Clearly, the probability that each output arm of the beam splitter contains at least one photon is given by  $1-2^{-(n-1)}$. We therefore have the following expectations: 1) The probability of simultaneous photon detections at both output arms (referred to as cross clicks) would increase with the photon number $n$; 2) in the limit of large $n$, the cross-click events would happen with near certainty. These motivate us to consider the associated cross-click operator 
\begin{equation}
F_{CC}=\one^{\textrm{A}}\otimes M_{CC}^\textrm{B},
\end{equation}
with  $M_{CC}^\textrm{B}$ being Bob's cross-click POVM element (see Appendix A and B of Ref.~\cite{Zhang2017} for the
derivation and expression of $M_{CC}^\textrm{B}$). To obtain bounds on the photon-number distribution using experimental observations, we thus consider the optimization problem 
\begin{equation}
\label{eq:min_cc}
\begin{array}{rc}
c_{n,min}=& {\min_{\rho_{\textrm{AB}}^{(n)}}}  \Tr\left(\rho_{\textrm{AB}}^{(n)}F_{CC}^{(n)}\right) \\
% \left(c_n\equiv \Tr\left(\rho_{\textrm{AB}}^{(n)}F_{CC}^{(n)}\right)\right) \\
 \text{subject to}&  \rho_{\textrm{AB}}^{(n)}\geq 0   \\
& \Tr\left(\rho_{\textrm{AB}}^{(n)}\right)=1  \; .
\end{array}
\end{equation} 
Here  $F_{CC}^{(n)}$ is the $n$-photon component  of the cross-click operator $F_{CC}$. 

Again, we solve this optimization problem using the YALMIP toolbox~\cite{YALMIP} in MATLAB. 
The numerical solutions of the optimization problem in Eq.~\eqref{eq:min_cc} provide 
strong evidence that the cross-click probability $c_{n,min}$ increases monotonically with $n$
and converges to the unit value $1$ for an arbitrary passive-detection efficiency mismatch.  
We would like to point out that any evaluation of secret-key rates using our approach requires 
solving an SDP problem, such as those in Eqs.~\eqref{eq:min_dc},~\eqref{eq:min_error} and~\ref{eq:min_cc}, 
thus allowing the validification of the working assumption for a chosen mismatch model and parameters. 
Particularly, the numerical evidence for our mismatch model and parameters is shown in 
Fig.~\ref{passive_crossclick}, which suggests the following two inequalities 
\begin{equation}\label{eq:crossclick_rela1}
c_{n,min}\geq c_{3,min}, \forall n\geq 3,
\end{equation}
and 
\begin{equation}\label{eq:crossclick_rela2}
c_{n,min}\geq c_{2,min}, \forall n\geq 2.
\end{equation}

\begin{figure}[htb!]
\includegraphics[scale=0.53,viewport=3cm 9cm 19cm 19.5cm]{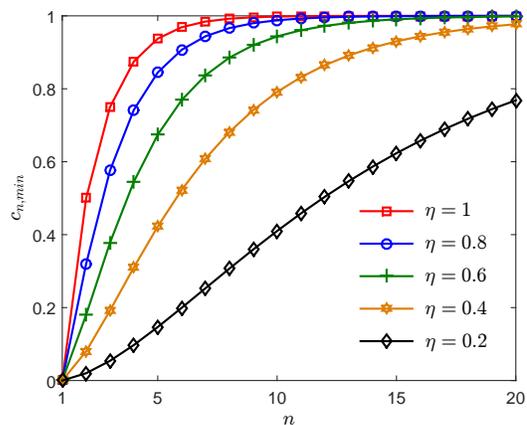}  
\caption{The minimum cross-click probability $c_{n,min}$ vs. the photon number 
 $n$ received by Bob for the passive-detection mismatch model of 
 Table~\ref{passive_mismatch} with $\eta_1=1$ and $\eta_2=\eta$. Note the monotonicity 
 of each curve as a function of $n$, supporting the inequalities in  
 Eqs.~\eqref{eq:crossclick_rela1} and~\eqref{eq:crossclick_rela2}.}
   \label{passive_crossclick} 
\end{figure}

The inequality in Eq.~\eqref{eq:crossclick_rela1} tells us that the cross-click 
probability $c_{\rm{obs}}$  observed in practice satisfies 
\begin{equation}
c_{\text{obs}}=\sum_{n=0}^{\infty} p_n  \Tr\left(\rho_{\textrm{AB}}^{(n)}F_{CC}^{(n)}\right)\geq \left( 1-p_0-p_1-p_2 \right)c_{3,min}. 
\label{eq:actual_cross_bound}
\end{equation}
Here we used the fact that $\sum_{n=0}^{\infty}p_n=1$. Thus we obtain a bound 
$b_2\leq(p_0+p_1+p_2)$ as  
\begin{equation}\label{eq:passive_b2_bound}
b_2=1-\frac{c_\text{obs}}{c_{3,min}}.
\end{equation}
Similarly, from Eq.~\eqref{eq:crossclick_rela2} we can obtain a bound $b_1\leq (p_0+p_1)$ as
\begin{equation}\label{eq:passive_b1_bound}
b_1=1-\frac{c_\text{obs}}{c_{2,min}}.
\end{equation}

The above bounds $b_1$ and $b_2$ together with the flag-state squasher approach
for the corresponding subspaces can be used in the key-rate optimization problem 
of Eq.~\eqref{eq:key_minimization} when the passive-detection scheme is used.

\section{Secret-key rates  with simulated observations} 
\label{sect:results}
As pointed out before, the method developed in Sect.~\ref{sect:method} allows a security analysis of a QKD setup with an arbitrary detection-efficiency mismatch. Any such security analysis requires the determination of constraints on the probability of the state in a subspace containing at most a given number of photons, and then a reliable key-rate lower bound can be obtained using those constraints together with a flag-state squasher. We now illustrate our approach for the specific mismatch models of Tables~\ref{active_mismatch} and \ref{passive_mismatch}. As the security analysis usually requires as input some data observed in experiments, we replace here the experiments by simulations according to a simple quantum-optical model. We specify this toy model below, but it is important to point out that this toy model is not part of the security analysis, or in anyway an assumption on which our security proof itself is based. We also  emphasize that 
the numerical values for the key rate reported in this section are reliable in the sense that they are computed 
according to the lower bound $K_{\text{lb}}$ on the key rate (see Eq.~\eqref{eq:key_lower_bound}). 

\subsection{Data simulation}
\label{sect:toy_model}
We study a BB84 protocol with an ideal single-photon source using polarization encoding.  As described in 
Sect.~\ref{sect:exp_con}, at each round of the protocol Alice prepares one of four possible single-photon 
polarization states selected uniformly randomly. Bob can use either the active- or passive-detection scheme. 
In the active-detection scheme, we assume that at each round Bob can randomly select the key-generation basis 
with probability $p=1/2$.  The single photon prepared by Alice is transmitted through Eve's domain to Bob.  
We model the  corresponding quantum channel as a depolarizing channel 
 $\Lambda(\rho) = (1-\omega) \rho + \omega \frac{1}{2} \openone$  with depolarizing probability $\omega$;
  additionally, the single-photon transmission efficiency over the channel is $t$. In order to introduce multiple
  detector clicks, Eve intercepts in our channel model with probability $r$ the single photon and resends 
  multiple photons to Bob.  Specifically, Eve resends randomly polarized $m$ photons in the state 
\begin{equation}
\rho_m=\frac{1}{2m\pi}\int_0^{2\pi} \mathrm{d}\theta \left(\hat{a}_\theta^\dagger\right)^m \ket{0}\bra{0} \left(\hat{a}_\theta\right)^m.
\label{sim_state}
\end{equation}
Here, the photon-creation operator $\hat{a}_\theta^\dagger$ is given in terms of the operators  
$\hat{a}_H^\dagger$ and $\hat{a}_V^\dagger$ of the respective linear polarizations as
$\hat{a}_\theta^\dagger=\cos(\theta) \hat{a}_H^\dagger+\sin(\theta)\hat{a}_V^\dagger$. 
In our simulations, we will choose the photon number $m=2$. 

When applying the flag-state squasher approach, we  choose to separate  
either the $(n\leq 1)$-photon or the  $(n\leq 2)$-photon subspace from their respective complements. 
That is, we set the photon-number flag threshold to be $k=1$ or $2$. 
In our efficiency-mismatch models we consider several spatial-temporal modes, in addition to 
the polarization mode. We note that the detectors used are assumed to be free of dark counts. 
 In our toy quantum channel, we additionally assume that the optical 
signals are uniformly randomly distributed over all considered spatial-temporal modes.

\subsection{Key rates in the absence of mismatch: Trade-offs between transmission efficiency and detection efficiency}
As mentioned in the introduction, when there is no efficiency mismatch between the detectors used 
in the measurement device, one can pull the detection inefficiency out of the detectors and into the 
channel action, creating an effective transmission loss. Consequently, the measurement device now 
is described by an ideal-detector setup for which a squashing model~\cite{normand2008, Kiyoshi2008, Gittsovich2014}
 exists, and so one can execute a full security proof. However, the resulting key rate might be conservatively low, because the existing security proof assumes that the photon loss during the actual transmission, as well as that due to the detection inefficiency, can be manipulated by Eve while under the original description of Bob's measurement device the photon loss inside of the device cannot be accessed by Eve. Such fact has been explicitly pointed out in literature such as in Ref.~\cite{Curty2004}. So while it is known that this is an overly pessimistic assumption, the issue is that proof techniques were missing to treat the security assuming the detection efficiency 
to be not accessible by Eve. We can tackle this question now with the techniques developed in this work. 

With our numerical method, we can prove the security of a QKD protocol with
arbitrary measurement operators as long as they are well characterized. In particular,
we can characterize the detection efficiency of each detector in a measurement device, 
and so we can determine the corresponding measurement operators (see Appendix A and B of Ref.~\cite{Zhang2017}). 
% doesn't rely the relationship and we can consider general measurement operators....
In this way,  we can study the individual effects of transmission efficiency and detection 
efficiency on the secret-key rate. To demonstrate these effects, for this particular result we assume for simplicity that each optical signal arriving at Bob contains no more than two photons, rather than using our flag-state squasher  approach. 

The results % \YZ{reliable key-rate lower bounds obtained by our numerical method} 
are shown in Fig.~\ref{trade-off}. From this figure, one can see that given 
the fixed total photon loss over both transmission and detection, Alice and Bob 
can distill more secret keys if they consider detection inefficiency and transmission loss separately  
rather than lumping these two kinds of loss together in the security proof. In particular,
when the product $t\eta$ is fixed,  the higher the value of $t$, the higher the secret-key rate is. 
On the other hand, when $t$ and $\eta$ are lumped together as an effective transmission efficiency 
$t\eta$, our numerical method provides the same key-rate lower bound as the standard security proofs 
with the help of the squashing model~\cite{normand2008,Kiyoshi2008,Gittsovich2014} for treating 
multiple-detection events. Specifically, the key-rate lower bound $\frac{1}{4}p_{\text{det}}\left(1-2 h(e)\right)$,
where $p_{\text{det}}$ is the detection probability at the key-generation basis, $e$ is the qubit error rate
and $h(e)$ is the binary entropy function,  is satisfied by the results plotted in 
Fig.~\ref{trade-off} when $\eta=1$. 

\begin{figure}[htb!]
\includegraphics[scale=0.53,viewport=2.1cm 4.8cm 19cm 16.1cm]{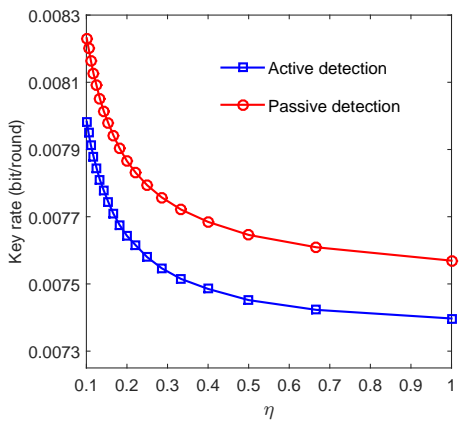}  
\caption{Reliable key-rate lower bound in bits per round  
obtained by our numerical method vs. the detection efficiency $\eta$ of all detectors used in
Bob's measurement device as shown in Fig.~\ref{setup}. We consider both the active- and 
passive-detection schemes. For data simulation, we fix the depolarizing 
probability $\omega=0.05$, the multi-photon probability $r=0.05$, and the 
product of transmission efficiency $t$ and detection efficiency $\eta$ 
to be $t\eta=0.1$. We choose these values just for ease of graphical 
illustrations. We remark that under each detection scheme, the 
probability distribution observed by Alice and Bob does not change 
with $\eta$ as long as the simulation parameters $\omega$, $r$ and $t\eta$
are fixed.}
\label{trade-off} 
\end{figure}

We also performed numerical calculations, not presented here, which show that the higher the 
multi-photon probability $r$, the more significant improvement in the secret-key rate is achieved 
when separating $t$ and $\eta$ in the security proof. Particularly, we observed that when the 
optical signal has no multi-photon component (i.e., $r=0$), the secret-key rate is independent 
of $\eta$ as long as $\omega$ and $t\eta$ are fixed. However, in practice multiple-detection 
events occur due to the use of sources containing multi-photon states, cross talks in fibers, 
or dark counts in detectors.

\subsection{Key rates with active-detection efficiency mismatch}
Let us study the dependence of the secret-key rate on the detection-efficiency mismatch with the 
active-detection scheme. We consider two scenarios: In the {\em one-mode} scenario all photons received 
by Bob are in the same spatial-temporal mode, and the two detectors labelled by `$H/D$' and `$V/A$' in 
Fig.~\ref{setup}(a) have efficiencies $\eta_1$ and $\eta_2$ respectively; in the {\em two-mode} scenario 
the photons received by Bob can stay in one of two possible spatial-temporal modes or in a 
coherent superposition of the two spatial-temporal modes.  
The efficiency mismatch for the combinations of spatial-temporal modes and polarization 
detectors is shown as in Table~\ref{active_mismatch}.  For security proofs, we make use of and 
compare two different assumptions/techniques to deal with potential multi-photon signals arriving 
at Bob's detectors: we either assume that each signal received by Bob contains no more than 
two photons, or we  prove security without such assumption. In the latter case we apply a flag-state 
squasher with the photon-number flag threshold $k=2$, and in the key-rate optimization problem of 
Eq.~\eqref{eq:key_minimization} we incorporate 
the lower bounds $b_1$ and $b_2$ on the photon-number probabilities $(p_0+p_1)$ and $(p_0+p_1+p_2)$. 
These bounds are based on observations and are discussed in Sect.~\ref{sect:photon_number_bounds} 
(see Eqs.~\eqref{eq:active_b2_bound} and~\eqref{eq:active_b1_bound}). 

\begin{figure}[hb!t]
   \includegraphics[scale=0.53, viewport=2.5cm 5cm 17.5cm 15.5cm]{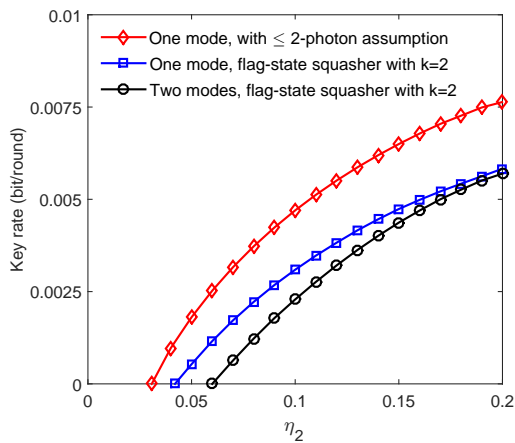}  
   \caption{ Reliable key-rate lower bound in bits per round 
   obtained by our numerical method vs. the detection efficiency $\eta_2$ of the detector   
   labelled by `$V/A$' (for the signals stayed in the first spatial-temporal mode) in the active-detection 
   scheme of Fig.~\ref{setup}(a). For data simulation, we fix the detection efficiency of the detector 
   labelled by `$H/D$' (for the signals stayed in the first spatial-temporal mode) to $\eta_1=0.2$. 
   We also fix the depolarizing probability $\omega=0.05$, the multi-photon probability $r=0.05$, 
   and the transmission efficiency $t=0.5$ (corresponding to $3$dB loss). We remark that for the 
   active-detection scheme the key rate scales linearly with the probability $p$ for Bob to select 
   the key-generation basis when other simulation parameters are fixed, and that for the results
   shown in this figure and the following Fig.~\ref{fig:active_distance} the probability $p$ is 
   fixed to be $1/2$ according to the data-simulation model detailed in Sect.~\ref{sect:toy_model}.}    
	\label{fig:active_mismatch}
\end{figure}

The typical results are %\YZ{dependence of the numerical key-rate lower bound on efficiency mismatch} 
 shown in Fig.~\ref{fig:active_mismatch}. We can make directly several observations 
from Fig.~\ref{fig:active_mismatch}:
\begin{enumerate}
\item The larger the efficiency mismatch, the lower the secret-key rate is. There exists a threshold for the efficiency mismatch beyond which it is not possible for Alice and Bob to distill secret keys.
\item Making assumptions on Eve's attack strategy, such as assuming that no more than two photons are being 
resent from Eve to Bob, can overestimate the true secret-key rate computed according to the analysis without making that assumption. 
\item The spatial-temporal-mode-dependent mismatch helps Eve to attack the QKD system. Our results show that Eve's corresponding freedom to manipulate the detection efficiencies decreases the secret-key rate.
\item If there is no efficiency mismatch, then the secret-key rate does not differ whether we consider one or 
two spatial-temporal modes. Note that in this case the lower bounds $b_1$ and $b_2$ in Eqs.~\eqref{eq:active_b1_bound} and~\eqref{eq:active_b2_bound} are independent of the number of spatial-temporal modes, and so 
 is the key-rate optimization problem in Eq.~\eqref{eq:key_minimization}. 
\end{enumerate}

We can also study the dependence of the secret-key rate on the transmission efficiency or distance 
when fixing other data-simulation parameters.  For this, we assume that the transmission 
efficiency $t$ is determined by the transmission distance $L$ in kilometers according to
$t=10^{-L/50}$. Also, as mentioned in Sect.~\ref{sect:toy_model} we assume that the detectors used are
free of dark counts. The typical results as shown in Fig.~\ref{fig:active_distance} % illustrate that
suggest the following observations.  First, with the increase of the transmission distance $L$, 
the key-rate lower bound obtained decreases. Particularly, when the distance $L$ is large and 
in the absence of dark counts, the key-rate lower bound obtained decreases exponentially with 
the increase of $L$. Second, in the limit of large distance $L$, the key-rate lower bounds 
obtained under different efficiency-mismatch models or using different assumptions/techniques 
to handle multi-photon signals approach to each other. % 2) In the limit of large distance $L$, the 
%key-rate lower bound obtained converges to the same value, regardless of which specific 
%efficiency-mismatch model is considered or whether a photon-number cut-off is assumed 
%for security analysis.}

\begin{figure}[hb!t]
   \includegraphics[scale=0.53, viewport=2.5cm 8.5cm 17.5cm 20cm]{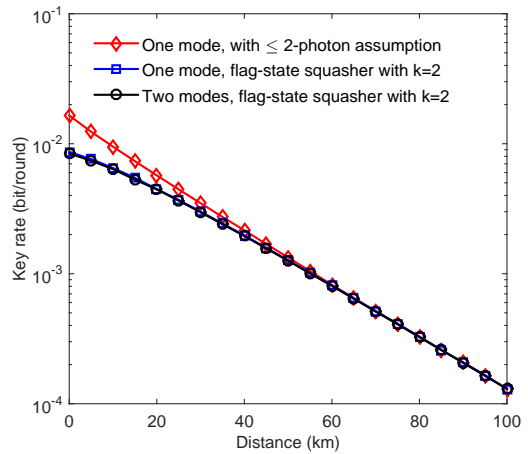}  
   \caption{Reliable key-rate lower bound in bits per round obtained by our numerical method 
   vs. the transmission distance in kilometers from Alice to Bob with the active-detection scheme 
   of Fig.~\ref{setup}(a).
   For data simulation, we fix the detection efficiencies of the two detectors labelled by `$H/D$' 
   and by `$V/A$' (for the signals stayed in the first spatial-temporal mode) 
   to $\eta_1=0.2$ and $\eta_2=0.18$, respectively.  We also fix the depolarizing probability 
   $\omega=0.05$ and the multi-photon probability $r=0.05$. Note that 
   % when the transmission distance  is small and 
   when the same photon-number flag threshold $k=2$ is used, 
   the key-rate lower bound obtained for the case of mode-dependent efficiency mismatch is slightly
   lower than that for the case of mode-independent efficiency mismatch, although due to the use of
   log-scale in the plot such a difference is hard to be visible.}  
  % labelled by blue squares and black circles for different mismatch models differ from each each  
  % We remark that for the 
%   active-detection scheme the key rate scales linearly with the probability $p$ for Bob to select 
%   the key-generation basis when other simulation parameters are fixed, and that for the results
%   shown here the probability $p$ is fixed to be $1/2$ according to the data simulation detailed
%   in Sect.~\ref{sect:toy_model}.   
	\label{fig:active_distance}
\end{figure}

\subsection{Key rates with passive-detection efficiency mismatch} 
As in the active-detection scheme, we consider two scenarios: In the {\em single-mode} scenario 
all photons received by Bob are in the same spatial-temporal mode, and the four detectors labelled 
by `$H$', `$V$', `$D$' and `$A$' in Fig.~\ref{setup}(b) have efficiencies 
$\eta_1$, $\eta_2$, $\eta_2$, $\eta_2$ respectively; in the {\em four-mode } scenario the  
photons received by Bob can stay in one of four possible spatial-temporal modes or in a 
coherent superposition of the four spatial-temporal modes. The efficiency mismatch 
in the four spatial-temporal modes is  shown as in Table~\ref{passive_mismatch}. 
In the security proofs,  we again compare the flag-state squasher approach with the photon-number 
cut-off assumption.  Note that for the case with one spatial-temporal mode, we apply a  
flag-state squasher with the photon-number flag threshold $k=2$, and at the same time 
we incorporate the lower bounds $b_1$ and $b_2$ in Eqs.~\eqref{eq:passive_b1_bound} 
and~\eqref{eq:passive_b2_bound}. For the case with four spatial-temporal modes, instead we apply 
a flag-state squasher with the photon-number flag threshold $k=1$, 
and exploit the corresponding photon-number distribution bound  $b_1$. 
We do not use the tighter approach with the larger photon-number flag threshold $k=2$,  
due to the large dimensionality of the corresponding key-rate optimization problem in the presence of four 
spatial-temporal modes.  The dependence of the secret-key rate
on the detection-efficiency mismatch is shown in Fig.~\ref{fig:passive_mismatch}. Similar to  
the active-detection case, the results in Fig.~\ref{fig:passive_mismatch} suggest that the larger 
the efficiency mismatch, the lower the secret-key rate is. When the efficiency mismatch is 
large enough, it is not possible for Alice and Bob to distill secret keys. 
The results also suggest that spatial-temporal-mode-dependent mismatch helps Eve to 
attack the QKD system.

\begin{figure}[hb!t]
   \includegraphics[scale=0.53, viewport=2.5cm 4.5cm 17.5cm 15.5cm]{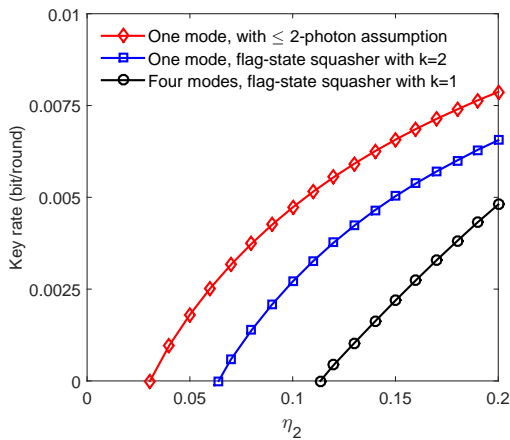}  
   \caption{Reliable key-rate lower bound in bits per round obtained by our numerical method vs. 
   the detection efficiency $\eta_2$ of the detectors   
   labelled by `$V$', `$D$' or `$A$' (for the signals stayed in the first spatial-temporal mode) in 
   the passive-detection scheme of Fig.~\ref{setup}(b). For data simulation, we fix the detection  efficiency 
   of the detector labelled by `$H$' (for the signals stayed in the first spatial-temporal mode) to 
   $\eta_1=0.2$. We also fix % the probability of selecting the key-generation basis by Bob $p=0.5$, 
   the depolarizing probability $\omega=0.05$, the multi-photon probability $r=0.05$, and the 
   transmission efficiency $t=0.5$ (corresponding to $3$dB loss).    
   }
	\label{fig:passive_mismatch}
\end{figure}

We remark that one cannot straightforwardly compare the robustness of the active- and passive-detection 
schemes against efficiency mismatch for distilling secret keys via Figs.~\ref{fig:active_mismatch} 
and~\ref{fig:passive_mismatch}. The reasons are as follows: First, there is no one-to-one 
correspondence between the two mismatch models given in Tables~\ref{active_mismatch} 
and~\ref{passive_mismatch}, for the active- and passive-detection schemes respectively. 
Second, for spatial-temporal-mode-dependent mismatch, in the active-detection scheme we 
considered two spatial-temporal modes and used the photon-number flag threshold $k=2$ as well as 
 the corresponding lower bounds on the photon-number probabilities $(p_0+p_1)$ 
and $(p_0+p_1+p_2)$. However, in the passive-detection scheme we considered four 
spatial-temporal modes and used the smaller photon-number flag threshold $k=1$ 
as well as the corresponding lower bound on the photon-number probability $(p_0+p_1)$. 
The higher the photon-number flag threshold and the more constraints on the photon-number 
distribution, the higher the secret-key rate certified by our method is. We emphasize  
that here we have developed a general method for proving security of practical 
QKD protocols with efficiency mismatch.  How to optimize our method and 
improve the secret-key rates certified will require future study.

We can also study the dependence of the secret-key rate on the transmission distance.  
Similar to the active-detection case, the typical results as shown in Fig.~\ref{fig:passive_distance} % illustrate that
suggest the following observations. First, in the limit of large transmission distance $L$, 
the key-rate lower bounds obtained under different efficiency-mismatch models or using 
different assumptions/techniques to handle multi-photon signals approach to each other.
Second, when the transmission distance $L$ is large, the key-rate lower bound obtained 
decreases exponentially with the increase of $L$. We note that when the distance $L$ is 
small, the key-rate lower bound obtained by the flag-state squasher approach with the 
photon-number flag threshold $k=1$ depends on $L$ in a non-monotonic way.  Such 
non-monotonic behaviour is understandable considering the following two competing facts: 
1) With the increase of $L$, the cross-click probability decreases and so the 
lower bound on the photon-number probability $(p_0+p_1)$ increases, which is 
helpful for our numerical method to distill secret keys; 2) with the increase of $L$, 
the detection probability decreases, which would result in a decrease of the key rate.
By using the larger photon-number flag threshold $k=2$, the above non-monotonic behaviour
disappears as we verified for the case with one spatial-temporal mode, see Fig.~\ref{fig:passive_distance}.  
% 2) In the limit of large distance $L$, the 
%key-rate lower bound obtained converges to the same value, regardless of which specific 
%efficiency-mismatch model is considered or whether a photon-number cut-off is assumed 
%for security analysis.}

\begin{figure}[hb!t]
   \includegraphics[scale=0.53, viewport=2.5cm 8.5cm 17.5cm 20cm]{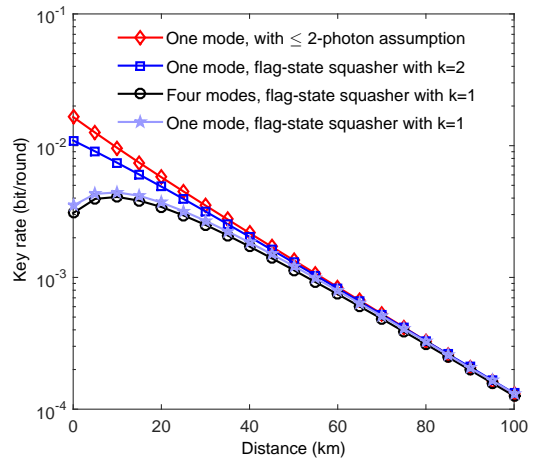}  
   \caption{Reliable key-rate lower bound in bits per round obtained by our numerical method 
   vs. the transmission distance in kilometers from Alice to Bob with the passive-detection scheme 
   of Fig.~\ref{setup}(b). For data simulation, we fix the detection efficiencies $\eta_1=0.2$
   and $\eta_2=0.18$, where $\eta_1$ and $\eta_2$ have the same meanings as those in 
   Fig.~\ref{fig:passive_mismatch}.  We also fix the depolarizing probability 
   $\omega=0.05$ and the multi-photon probability $r=0.05$. Note that 
   when the same photon-number flag threshold $k=1$ is used, 
   the key-rate lower bound obtained for the case of mode-dependent efficiency mismatch 
   is slightly lower than that for the case of mode-independent efficiency mismatch.
   Also, by comparing the results obtained using the photon-number flag threshold $k=1$ 
   with those obtained using $k=2$ in the case of one spatial-temporal mode, we can see
   that the usage of the photon-number flag threshold $k=1$ can induce a non-monotonic 
   behaviour of the obtained key-rate lower bound as a function of the transmission distance.}  
  % labelled by blue squares and black circles for different mismatch models differ from each each  
  % We remark that for the 
%   active-detection scheme the key rate scales linearly with the probability $p$ for Bob to select 
%   the key-generation basis when other simulation parameters are fixed, and that for the results
%   shown here the probability $p$ is fixed to be $1/2$ according to the data simulation detailed
%   in Sect.~\ref{sect:toy_model}.   
	\label{fig:passive_distance}
\end{figure}

\section{Conclusion}
\label{sect:conclusion}
The security proof of QKD usually assumes that the threshold detectors used
 have the same detection efficiency. However, in practice, their detection efficiencies 
can show a mismatch, either due to the manufacturing and setup, or the influence 
by  Eve (for example, by controlling the spatial-temporal-mode-dependent coupling 
of an optical signal with a detector).  In this work we present an approach that 
allows to lower-bound the secret-key rate of a QKD setup with an arbitrary, but 
characterized detection-efficiency mismatch. We formulate the key-rate calculation as a 
convex-optimization problem. In order to prove security without relying on a 
cut-off of photon numbers in the optical signal, we exploit the bounds 
on the photon-number distribution obtained directly from experimental observations 
 with the help of semi-definite programs (SDPs), 
and simplify the key-rate optimization problem by introducing a flag-state 
squashing map. The key-rate optimization 
problem formulated is based on the practical measurement operators that
depend on the characterized efficiency mismatch. Therefore, we can 
study the effect of efficiency mismatch on the secret-key rate. 

We illustrate the power of our method with numerical simulations, demonstrating 
that our method can be numerically well handled even in the presence of 
spatial-temporal-mode-dependent mismatch. Our method is especially applicable 
to free-space QKD where spatial-temporal-mode-dependent mismatch can be easily 
induced by Eve as demonstrated in Refs.~\cite{Rau2014, Shihan2015, Chaiwongkhot2019}.

Moreover, with our method, one can clearly see the individual effects of transmission loss
and detection inefficiency on the secret-key rate (see Fig.~\ref{trade-off}). 
In the particular case of no mismatch, the simulation results show that 
our method provides a tighter lower bound on the secret-key rate than 
the squashing model~\cite{normand2008, Gittsovich2014, Kiyoshi2008} when we 
separate detection inefficiency (out of the domain of Eve) from 
transmission loss (in the domain of Eve). \\

\noindent{\textbf{Note added.}} After the submission of our work, 
we noticed that a related work by Trushechkin appeared on arXiv, see 
Ref.~\cite{Trushechkin2020}.  In contrast to our numerical bounds on 
the photon-number distribution obtained by solving semidefinite programs,
Trushechkin~\cite{Trushechkin2020} derived analytical bounds for the 
active-detection case.  These analytical bounds can be combined  
with the flag-state squasher introduced in our work for a security 
proof without a cut-off of photon numbers in the optical signal.  
Motivated by Trushechkin's work~\cite{Trushechkin2020} and the 
construction of squashing models presented in Ref.~\cite{normand2008}, 
we can derive better analytical bounds on the photon-number distribution. 
 We will %for both the active- and the passive-detection schemes.
present the details of these analytical bounds and their applications
in the future work. \\
%In parallel with our work (detailed in Ref.~\cite{shalm:qc2019}), another two 
%works~\cite{li:qc2019, liu:qc2019} also demonstrate device-independent randomness expansion. 
%In contrast with our work, Refs.~\cite{li:qc2019, liu:qc2019} exploit the standard spot-checking 
%protocol, which requires both uniform and highly biased random bits as the protocol input. 
%On the other hand, both works perform security analysis against quantum side information. 
%With respect to proof techniques, Ref.~\cite{li:qc2019} and Ref.~\cite{liu:qc2019} employ 
%quantum probability estimation~\cite{zhang_y:2020b,knill:qc2018a} and entropy 
%accumulation~\cite{dupuis:qc2016a,arnon-friedman:qc2018a} respectively. 
%In particular, the system detection efficiency achieved in Ref.~\cite{li:qc2019} is high 
%enough ($>80\%$) such that the numerical implementation of quantum probability estimation 
%requires only a few hours.  

\begin{acknowledgments}
We thank Shihan Sajeed, Poompong Chaiwongkhot, and Vadim Makarov
for many useful discussions and comments.  We gratefully acknowledge supports
through the Office of Naval Research (ONR), the Ontario Research Fund (ORF), 
the Natural Sciences and Engineering Research Council of Canada (NSERC), and 
Industry Canada. Financial
support for this work has been partially provided
by Huawei Technologies Canada Co., Ltd.
\end{acknowledgments}

% \newpage
\appendix
\section{Alice's measurement operators}
\label{sect:Alice_operators}
In the source-replacement description of a BB84-QKD protocol with an ideal single-photon source using 
polarization encoding,  the quantum system A held by Alice is two-dimensional and Alice performs 
ideal one-qubit measurements with perfect detection efficiency. In particular, Alice's 
measurement operators $M_H^{\textrm{A}}$, $M_V^{\textrm{A}}$, $M_D^{\textrm{A}}$, and 
$M_A^{\textrm{A}}$ are ideal polarization-measurement operators in the one-photon space. 
In the basis $\{\Ket{1_H,0_V}_{\textrm{A}}, \Ket{0_H,1_V}_{\textrm{A}}\}$ 
of Alice's one-photon space over two polarization modes, these operators are represented as
\begin{align}
 & M_H^{\textrm{A}}=\left(\begin{array}{c c}
       1 &   0     \\
       0 &   0     
     \end{array}\right), \hspace{0.8cm}
 M_V^{\textrm{A}}=\left(\begin{array}{c c}
       0 &   0     \\
       0 &   1     
     \end{array} \right),     \notag \\
 & M_D^{\textrm{A}}=1/2\left(\begin{array}{c c}
       1 &   1     \\
       1 &   1     
     \end{array} \right), \hspace{0.2cm} % \text{ and }
 M_A^{\textrm{A}}=1/2\left(\begin{array}{c c}
       1 &   -1     \\
       -1 &   1     
     \end{array} \right).      
     \label{eq:pauli_xz} 
\end{align}

% \bibliography{Security_proof_with_mismatch}

\end{document}